\documentclass[runningheads]{llncs}

\usepackage{amsmath,amssymb,amsfonts}
\usepackage[colorlinks]{hyperref}
\hypersetup{citecolor=blue}
\usepackage{graphicx}
\usepackage{comment}
\usepackage{color}
\usepackage[linesnumbered, ruled]{algorithm2e}
\usepackage{calc} 

\usepackage{graphicx}

\def\calD{\mathcal{D}}

\def\calL{\mathcal{L}}
\def\calS{\mathcal{S}}

\def\calO{\mathcal{O}}
\def\calR{\mathcal{R}}

\newcommand{\R}{\mathbb{R}}
\newcommand{\N}{\mathbb{N}}

\newcommand{\Z}{\mathbb{Z}}

\begin{document}

\title{{Online hitting set of $d$-dimensional fat objects}}
\author{Shanli Alefkhani\inst{1}\and
Nima Khodaveisi\inst{1} \and
Mathieu Mari\inst{1,2}}

\authorrunning{S. Alefkhani et al.}
\institute{IDEAS-NCBR, Warsaw, Poland \and
University of Warsaw, Poland}

\maketitle
\begin{abstract}
    We consider an online version of the geometric minimum hitting set problem that can be described as a game between an adversary and an algorithm. For some integers $d$ and $N$, let $P$ be the set of points in $(0, N)^d$ with integral coordinates, and let $\calO$ be a family of subsets of $P$, called objects. Both $P$ and $\calO$ are known in advance by the algorithm and by the adversary. Then, the adversary gives some objects one by one, and the algorithm has to maintain a valid hitting set for these objects using points from $P$, with an immediate and irrevocable decision. We measure the performance of the algorithm by its competitive ratio, that is the ratio between the number of points used by the algorithm and the offline minimum hitting set for the sub-sequence of objects chosen by the adversary.
    
    We present a simple deterministic online algorithm with competitive ratio $((4\alpha+1)^{2d}\log N)$ when objects correspond to a family of $\alpha$-fat objects. Informally, $\alpha$-fatness measures how cube-like is an object. We show that no algorithm can achieve a better ratio when $\alpha$ and $d$ are fixed constants. In particular, our algorithm works for two-dimensional disks and $d$-cubes which answers two open questions from related previous papers in the special case where the set of points corresponds to all the points of integral coordinates with a fixed $d$-cube.

\end{abstract}

\section{Introduction}
The hitting set problem is one of the fundamental problems in combinatorial optimization. Let $(X, \calR)$ be a range space where $X$ is a set of elements and $\calR$ is a family of subsets of $X$, $|X| = n, |\calR| = m$. A subset $H \subseteq X$ is a \emph{hitting set} for $\calR$ if and only if, for every range $R \in \calR$ the intersection of $H$ and $R$ is non-empty. 
In the offline setting, the goal is to find a hitting set of minimum size. Note that by interchanging the roles of subsets and elements, the hitting set problems is equivalent to the set cover problem. 
The hitting set problem is a classic NP-hard problem \cite{Karp72}, and the best approximation factor achievable in polynomial time (assuming $P \neq NP$) is $\Theta(\log n)$ \cite{Chvatal79,Feige98,Johnson74a,Lovasz75}. 

There is a line of work that considered the hitting set problem in a geometrical setting. The set of elements $X$ is a subset of \emph{points} of the $d$-dimensional plane $\R^d$ and $\calR$ corresponds to a family of geometrical objects\footnote{to simplify, we consider that a subset of points $D\subseteq X\subseteq \R^2$ is a disk (or a square, or another type of geometric object) if there exists a disk $D'\subset \R^2$ such that $D'\cap X=D$. This allows us to consider a subset of $X$ as a geometrical object. }, e.g., disks, squares, rectangles, etc., for $d=2$. The hitting set problem remains NP-hard even for simple geometric objects like unit disks or unit squares in $\R^2$ \cite{FowlerPT81}. For some families of geometric objects, there are better approximation ratios than for the general case, e.g., a PTAS for axis-parallel squares and disks \cite{MustafaR09}, and more generally for fat objects in a fixed dimension \cite{CHAN2003}.

In this paper, we consider an \emph{online} version of the problem. It is convenient to define this problem as a game between an adversary and an algorithm. Initially, a range space $(X,\calR)$ is known in advance by both the algorithm and the adversary. The game consists of a series of turns until the adversary decides to stop the game. In each turn, the adversary gives a subset $R\in \calR$, and the algorithm has to choose a point $p\in X$, such that $p\in R$, if none of the points previously chosen by the algorithm are contained in $R$. The algorithm is allowed to select several points during the same turn and may decide to select new points even though the current subset $S$ is already hit by one of its previous points. The goal of the algorithm is to minimize the total number of points selected at the end of the game. 
See Figure \ref{fig:online hitting set} for an illustration of this game. 
We measure the performance of an algorithm by its \emph{competitive ratio}, which corresponds to the ratio between the number of points selected by the algorithm and the minimum size of a hitting set of the sub-family of $\calR$ given by the adversary during the game. 

\begin{figure}[h!]
    \centering
    \includegraphics[page=1, width=0.5\textwidth]{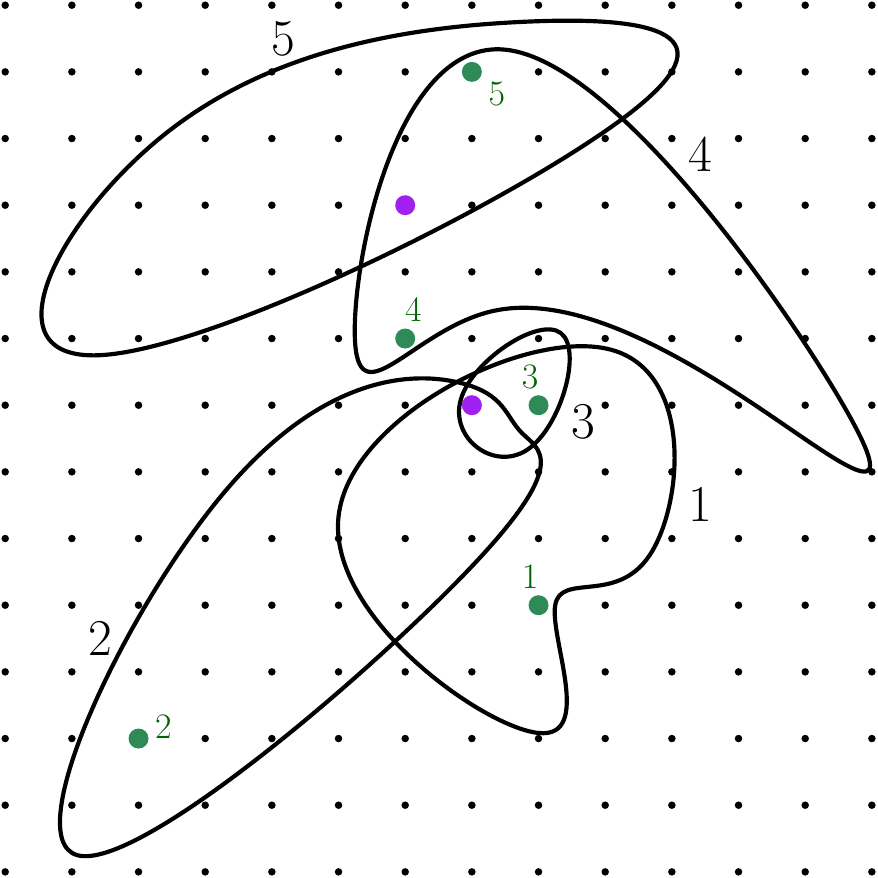}
    \caption{Here, the set of points $X$ corresponds to the points of integral coordinates of the Euclidean plane (small black dots). During the game, the adversary has given 5 subsets. In each turn, the algorithm has chosen one point (green) that is contained in the object given by the adversary. In total, the algorithm has used 5 points while the offline minimum hitting set consists of only two points (purple). }
    \label{fig:online hitting set}
\end{figure}

In the online setting, Alon et al. \cite{AlonAABN09} introduced an (essentially) tight $O(\log n\log m)$-competitive algorithm for the general case. There are also a few works that considered the special case of geometrical objects. 
 Even et al. \cite{EvenS11} presented an $O(\log n)$-competitive algorithm for intervals ($d = 1$), for unit-disks ($d=2$), and later for half-planes in dimension two \cite{EvenS14}. Khan et al. \cite{Khan+23} presented $O(\log N)$-competitive algorithm for axis-parallel squares of arbitrary sizes, assuming that all points have integral coordinates in $[0, N)^2$. 
 De et al. \cite{DeS22} looked at the problem in dimension $d$, when the algorithm is allowed to use any point of integral coordinates. They showed an $O(d^2)$-competitive algorithm for unit hypercubes and an $O(d^4)$-competitive algorithm for unit balls in dimension $d$. They also showed that any deterministic online algorithm for hypercubes has a competitive ratio of at least $d+1$.  Even and Smorodinsky also showed a lower bound of $\Omega(\log n)$ for intervals and arbitrary points \cite{EvenS11}. 
 
 In this paper, we are interested in two open questions mentioned in these papers:  
\begin{enumerate}
    \item Can one obtain an $o(\log^2 n)$-competitive algorithm for disks (dimension two)? \cite{EvenS14}
    \item Can one obtain an $o(\log n\log m)$-competitive algorithm for cubes (dimension three)? \cite{Khan+23}
\end{enumerate}

\subsection{Our contribution}
In this paper, we are interested in $d$-dimensional fat objects, that generalize disks, squares, hypercubes, etc. An object $O\in \R^d$ is $\alpha$-fat, for some $\alpha\ge 1$ if the ratio of the sizes of the smallest hypercube containing $O$ and the biggest hypercube contained in $O$ is at most $\alpha$. 

We answer the two open questions mentioned above in the case where $X$ corresponds to the set of points with integral coordinates that are contained in a fixed hypercube. More precisely, let $d$ and $N$ be some integers. Let $P=(0,N)^d\cap \Z^d$ be the set of points of integral coordinates that are contained in $(0, N)^d$. In particular $n=|P|=(N-1)^d$. Let $\calO$ be a family of subsets of $P$. 

\begin{theorem}
There is an $((4\alpha+1)^{2d} \log N)$-competitive algorithm for minimum hitting set on $(P,\calO)$ when $\calO$ corresponds to a family of $d$-dimensional $\alpha$-fat objects in $(0, N)^d$. 
\end{theorem}

Notice that disks are $\sqrt{2}$-fat and cubes are 1-fat. Thus, Theorem 1 settles both questions from the introduction in the affirmative.

This algorithm is  $O(\log n)$-competitive ratio for disks of arbitrary sizes in the $2$-dimensional plane, and $3$-dimensional cubes. 

Our algorithm works as follows. It associates to each point in $P$ a color in $\{0,\dots, \lfloor \log N \rfloor -1\}$. Then, when an object arrives, if it is already hit we do nothing, otherwise we pick all the points with the maximum color inside the object.
Our coloring guarantees that we add at most $(4\alpha+1)^d$ points in each step.  We show that for each color $l$, and each point $p$ in the offline solution, the adversary cannot give more than $(4\alpha+1)^d$ objects that are not already hit at their arrival time, that contain $p$ and are of level\footnote{The level of an object is the maximum color of the points contained inside it.} $l$. This will help us to prove our competitive ratio.

We also show a lower bound of $\Omega(\frac{\log N}{1 + \log \alpha})$ on the competitiveness of any algorithm for the problem. This implies that no algorithm can achieve a better ratio when $\alpha$ and $d$ are fixed constants.

\section{The algorithm}
In this section, we present our online algorithm for hitting set of $d$-dimensional $\alpha$-fat objects. We start with some useful definitions. Let $d$ and $N$ be two integers and $P$ be the set of points with integral coordinates in $(0, N)^d$. 
A $d$-cube is an axis-parallel $d$-dimensional hypercube. The \emph{width} of a $d$-cube is the length of any of its sides. For simplicity, we assume that all geometrical objects $O\subset \R^d$ considered in this paper are open sets. 

\begin{definition}[$\alpha$-fat]     
   Let $O\subset \R^d$. The \emph{in-width} of $O$ is the length of the largest $d$-cube contained in $O$. The \emph{out-width} of $O$ is the length of the smallest $d$-cube containing $O$. We say that $O$ is $\alpha$-\emph{fat}, for some $\alpha \geq 1$ if the ratio of its out-width over its in-width is at most $\alpha$. 
   \label{definition:fat}
\end{definition}

For instance $d$-disks are $O(\sqrt{d})$-fat. Also, notice that $d$-cubes are $1$-fat.

\begin{definition}[level of a point]
   Let $i\in \N$, we denote $\ell(i)$ the maximum number $k$ such that $i$ is a multiple of $2^k$. For a point $x=(x_1,\dots, x_d)\in P$, we define the level of a point to be $\ell(x)=\min_{i=1}^d \ell(x_i)$. 
\end{definition}

See figure \ref{fig:coloring} for an illustration of the levels of the points of $P$. We denote $\calL=\{0,\dots, \lfloor\log N\rfloor - 1 \}$. It is clear that for each point $x\in P$, we have $\ell(x)\in \calL$. 

\begin{figure}[t]
    \centering
    \includegraphics[page=2, width=0.5\textwidth]{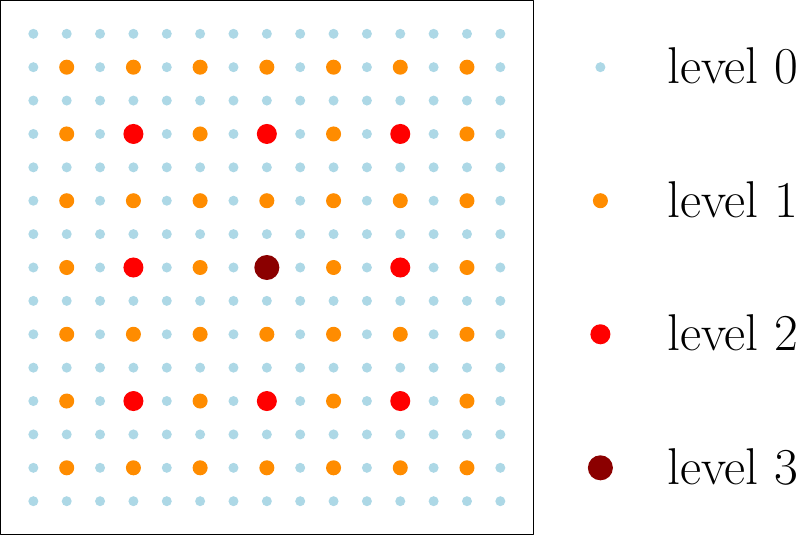}
    \caption{Levels of points in $P$ for $N=16$.}
    \label{fig:coloring}
\end{figure}

\begin{definition}[level of an object]
    For a geometric object $O \subseteq (0, N)^d$, we define its \emph{level} $\ell(O)$ as the maximum level over all the points in $O \cap P$. For each $l \in \calL$, we denote $n_l(O)$ as the number of points of level $l$ in $O \cap P$ and $n_{\ge l}(O)$ the number of points of level at least $l$ in $O \cap P$.
\end{definition}

We now describe the algorithm. We maintain a hitting set $P'$ that is initially empty.
In each round, we are given an $\alpha$-fat object $O \subseteq (0, N)^d$. 
\begin{itemize}
    \item If $O$ is already hit by a point in $P'$, then we do nothing.
    \item Otherwise, we add all the points in $P\cap O$ that are of level $\ell(O)$ to $P'$.
\end{itemize}
     
It is clear that at the end of each turn, $P'$ is a hitting set of the objects given so far. Notice that in each step, the action made by the algorithm only depends on $P'$ and the current object, but not on the previous objects given by the adversary. 

Now we prove that the competitive ratio of this algorithm is $((4\alpha+1)^{2d}\log N)$. 

\begin{lemma}
    The in-width of an object $O \subseteq (0, N)^d$ is less than $2^{\ell(O)+1}$.  
     \label{lemma:geo_at_least_d}
\end{lemma}

See figure \ref{fig:lemma1} for an example. 

\begin{proof}[Proof of Lemma \ref{lemma:geo_at_least_d}.]
    Let $l=\ell(O)$. We define $S_O$ as the largest $d$-cube contained in $O$. 
    We assume for the sake of a contradiction that the width of $S_O$ is at least $2^{l+1}$ and we show that there exists a point $q\in S_O$ such that $\ell(q)\ge l+1$. 
    Since the width of $S_O$ is at least $2^{l+1}$, for each $i$, $1\le i\le d$, there exists an integer $k_i$, such that the hyperplane with $i$-th coordinate $k_i2^{l+1}$ intersects $S_O$. Then, the $q=(k_12^{l+1},\dots, k_d2^{l+1})$ is contained in $S_O$, and its level is $\ell(q)\ge l+1$. This proves the lemma.
\end{proof}

\begin{figure}[t]
    \centering
    \includegraphics[page=3, width=0.5\textwidth]{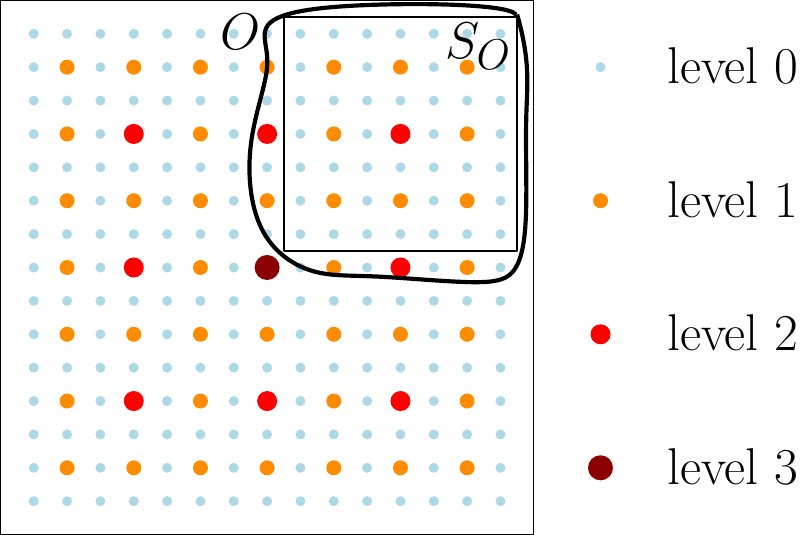}
    \caption{This figure shows an object $O$ of level $l=2$ in dimension two. $N=16$ and the in-width of $O$ is $7\le 2^{l+1}$.}
    \label{fig:lemma1}
\end{figure}

\begin{corollary}
    Let $O \subseteq (0, N)^d$ be an $\alpha$-fat object, for some $\alpha\ge1$. Then, the out-width of $O$ is at most $2^{\ell(O) + 1}\alpha$.
    \label{corollary:out_width_level}
\end{corollary}

\begin{lemma}
    Let $l \in \calL$ and $\alpha \geq 1$. A $d$-cube of width at most $\alpha 2^{l+2}$ contains at most $(4\alpha + 1)^d$ points of level $l$. 
    \label{lemma:geo_at_most_d}
\end{lemma}

See figure \ref{fig:lemma2} for an example.

\begin{proof}[Proof of Lemma \ref{lemma:geo_at_most_d}]
    Let $C \subseteq (0, N)^d$ be a $d$-cube of width $w \le \alpha 2^{l+2}$. 
    For each $i$, $1\le i\le d$, let $\Lambda_i$ be the set of integers $\lambda$ such that (i) $\lambda$ is a multiple of $2^l$ and, (ii) the hyperplane whose $i$-th coordinate is $\lambda$ intersects $C$. For each $i$, $1\le i\le d$, we have $|\Lambda_i|\le \frac{w}{2^l}+1\le  4\alpha+1$. 
    
    It is easy to see that $\prod_{i=1}^d \Lambda_i$ is the set of points of $C$ that are of level at least $l$. Thus, 
    $$n_l(C)\le n_{\ge l}(C)\le |\prod_{i=1}^d \Lambda_i|=\prod_{i=1}^d |\Lambda_i|\le (4\alpha+1)^d.$$
    This finishes the proof. 
    \end{proof}

\begin{figure}[t]
    \centering
    \includegraphics[page=4, width=0.5\textwidth]{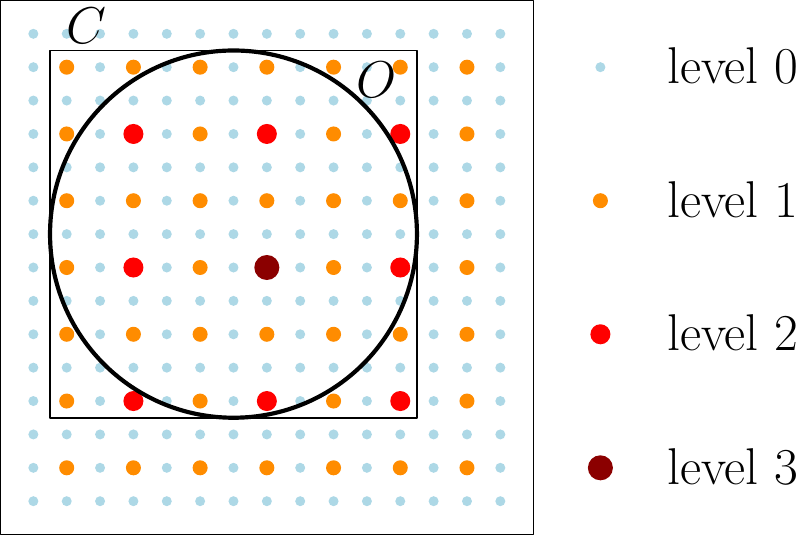}
    \caption{In this figure, $O$ is a two-dimensional disk. In particular, $O$ is $\alpha$-fat for $\alpha=\sqrt{2}$. Let $l=1$. The smallest $2$-cube containing $O$ is $C$, which has width less than $11\le \alpha 2^{l+2}\approx 11.31$. The number of points of level $l$ contained in $C$ is $27\le (4 \sqrt{2} + 1)^2\approx 44.31$. }
    \label{fig:lemma2}
\end{figure}

\begin{corollary}
    Let $O\subseteq (0,N)^d$ be an $\alpha$-fat object for some $\alpha\ge 1$. Then, $n_{\ell(O)}(O) \leq (4\alpha + 1)^d$.
    \label{corrolary:upper_bound_nl}
\end{corollary}

We now analyze the competitive ratio of our algorithm and show the following bound.

\begin{lemma}
Our algorithm is $((4\alpha+1)^{2d} \log N)$-competitive. 
\label{lemma:main}
\end{lemma}
\begin{proof}

Let $\calS\subseteq\calO$ denote the sequence of objects given by the adversary and $\calS'\subseteq \calS
$ denotes the sub-sequence of objects that are not already hit at their arrival. For each $O\in\calS$, our algorithm picks the points of level $\ell(O)$ inside $O$ if $O$ is not already hit, i.e., if $O\in\calS'$. By corollary \ref{corrolary:upper_bound_nl}, we know that $O$ contains at most $(4\alpha + 1)^d$ points of level $l$, and thus, our algorithm returns a hitting set of size 
\begin{equation}
    |P'| \leq |\calS'|(4\alpha + 1)^d
    \label{eq:bound_alg}
\end{equation}
We now establish an upper bound on the size of the minimum hitting set $\textsc{OPT}\subseteq P$ of $\calS$. 

For each $l\in \calL$ and each $p\in P$, we denote $\calS'_{l,p}\subseteq \calS'$ the set of objects of level $l$ in $\calS'$ that contain $p$.

    \begin{claim}
        $|\calS_{l, p}'| \leq (4\alpha + 1)^d$.
    \end{claim}
        Assume for the sake of a contradiction, that $|\calS_{l, p}'| > (4\alpha + 1)^d$. We denote $B_p$ as the $d$-cube of width $2^{l+2}\alpha$ centered in $p$. By Lemma \ref{corrolary:upper_bound_nl}, we know that there are at most $(4\alpha + 1)^d$ points of level $l$ in $B_p$. Also, it is clear by Corollary \ref{corollary:out_width_level} that any object of level $l$ containing $p$ is inside $B_p$. Then, by the pigeonhole principle, there are two objects $O, O' \in \calS_{l, p}'$ such that they both contain the same point $q$ of level $l$. See figure \ref{fig:upper-bound}. Without loss of generality, let us assume that $O$ arrived before $O'$ in the sequence of objects given by the adversary. Then, our algorithm picks $q$ when $O$ arrives since both $O$ and $q$ have level $l$. On the other side, we know that $q \in O'$, meaning that $O'$ is already hit when it is given by the adversary, which is a contradiction with the fact that $O'$ is in $\calS'$. Therefore, for each $l\in \calL$ and each $p\in P$, we have $|\calS_{l, p}'| \leq (4\alpha + 1)^d$.

\begin{figure}[t]
    \centering
    \includegraphics[page=5, width=0.5\textwidth]{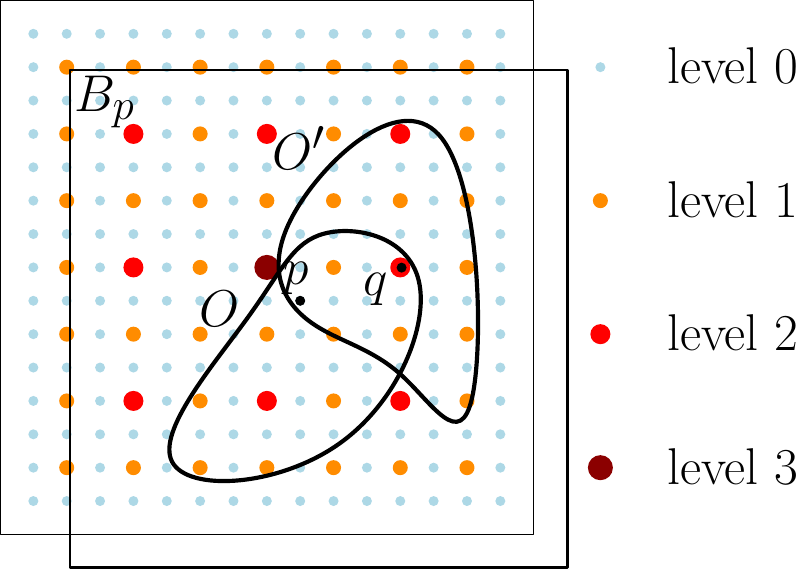}
    \caption{Here, $l=2$. We show in the proof of Lemma \ref{lemma:main} that the adversary gave two objects $O$ and $O'$ of level $l=2$, containing a common point $p$, that share a common point $q$ of level two. This cannot happen if $O$ and $O'$ are both not already hit at their arrival time.
    }
    \label{fig:upper-bound}
\end{figure}

     Since $\textsc{OPT}$ is a hitting set for $\calS$, it is also an hitting set for $\calS'$, which implies that 
     $$
     \calS'\subseteq \bigcup_{p\in\textsc{OPT}}\bigcup_{l\in\calL}\calS'_{l,p}.
     $$
     With the previous upper bound on the size of $\calS'_{l,p}$, we obtain that
     $$
     |\calS'|\le |\textsc{OPT}|\cdot|\calL|\cdot(4\alpha + 1)^d\le |\textsc{OPT}|\cdot(\lfloor\log N\rfloor)\cdot(4\alpha + 1)^d,
     $$
     which together with the bound of equation \eqref{eq:bound_alg} implies that our algorithm is $((4\alpha+1)^{2d} \log N)$-competitive. 
\end{proof}

\section{Lower Bound}
In this section, we prove the lower bound for the problem (Theorem \ref{theorem:lower_bound}). For any object $O \subset \R^d$, we say that object $O'\subset \R^d$ is a \emph{dilation} of $O$ if it is the result of a translation and a homothety of $O$ with positive scale factor. More formally, $O'$ is a dilation of $O$ if there exists $\beta > 0$ and translation vector $v$ such that $O' = \beta O + v$. Let $\calD(O)$  denote the set of dilations of $O$. 
Notice that if $O$ is $\alpha$-fat, for some $\alpha\ge1$, then a dilation of $O$ is also $\alpha$-fat. 

Recall that the online hitting set problem can be formalized as a game between an adversary and an algorithm.  

\begin{theorem}
    Consider the range space $(P,\calD(O))$, for any $\alpha$-fat object $O\subset \R^d$. The adversary 
has a strategy that forces the algorithm to place at least $\frac{\log N}{1 + \log \alpha}$ points, whereas the optimum offline solution only requires one point. 
\label{theorem:lower_bound}
\end{theorem}
\begin{proof}

The adversary produces a sequence $O_1, O_2, \dots, O_s \in \calD(O)$ as follows. In the first step, the adversary chooses $O_1\in\calD(O)$ such that the smallest $d$-cube enclosing $O_1$ is $(0,N)^d$. Then, for each $j \geq 1$, the adversary does the following. 
At step $j$, with $j\ge 2$, let $O_{j}$ be the largest dilation of $O$ that is contained in $O_{j-1}$ and that does not contain any point from the algorithm. 
If $O_{j}$ does not contain any point in $P$, the game ends. Otherwise, the adversary gives $O_j$ to the algorithm; See Figure \ref{fig:adversary_strategy}. 

Let $O_1, O_2, \dots, O_s$ be the sequence of objects obtained at the end of the game. Since $O_1\supseteq O_2\supseteq  \dots\supseteq  O_s$, and there exists a point $p\in P\cap O_s$, the set $\{p\}$ is a hitting a hitting set of $\{O_1, O_2, \dots, O_s\}$. 

Now we give a lower bound on the number of points used by the algorithm. 
For each $j$, $1\le j\le s$, let $k_j$ be the number of points added by the algorithm during step $j$. It is clear that the algorithm uses in total $\sum_{j=1}^sk_j$ points. We show that $\sum_{j=1}^sk_j\ge \frac{\log N}{1+\log \alpha}$. 

For each $j$, $1\le j\le s$, let $C_j$ ($C'_j$) be the smallest (largest, resp.) $d$-cube containing (contained in, resp.) $O_j$ and let $w_j$ ($w'_j$, resp.) denote its width. Notice that $w_j$ ($w'_j$) is the out-width (in-width, resp.) of $O_j$. See Figure \ref{fig:lower_bound}. We claim that $w_{j+1} \geq \frac{w_j}{\alpha(k_j+1)}$.

Let $P_j$ be the set of points from the algorithm that are contained in $C_{j}'$ at the end of step $j$. We have $|P_j|\le k_j$ since $O_j$ is not already hit at the beginning of step $j$, and the algorithm adds $k_j$ new points during that step. 
We prove that there exists a $d$-cube of width at least $\lfloor\frac{w_j'}{k_j+1}\rfloor$ inside $C_j'$ which does not contain any points of $P_j$. 
Let $p=(x_1,\dots,x_d)$ be a corner of $C_j'$ such that for any other point $p'=(x_1', \dots, x_d') \in C_j'$ and any $1\le i\le d$, we have $x_i \le x_i'$. For each $i$, $1 \le i \le d$, consider the set of $k_j+1$ intervals $\{ [x_i + h \frac{w_j'}{k_j+1} , x_i + (h+1) \frac{w_j'}{k_j+1}] \mid 0 \le h \le k_j, h\in\Z\}$. Since $|P_j| \le k_j$, there exists an integer $0 \le h_i \le k_j$ such that the $i$-th coordinate of every point in $P_j$ is not in $(x_i + h_i  \frac{w_j'}{k_j+1} , x_i + (h_i+1)  \frac{w_j'}{k_j+1} )$. 
Therefore, if we consider the $d$-cube $S_j=\prod_{i=1}^d [x_i + h_i  \frac{w_j'}{k_j+1} , x_i + (h_i+1)  \frac{w_j'}{k_j+1} ]$, $S_j$ does not contain any points of $P_j$, and the width of $S_j$ is $\frac{w_j'}{k_j+1}$ which implies that $w_{j+1}$ would be at least $\frac{w_j'}{k_j+1}$ (See Figure \ref{fig:lower_bound}). Also, note that $S_j$ is inside $C_j'$. Moreover, since $O_i$ is $\alpha$-fat, we have $w_{j+1} \ge \frac{w_j'}{k_j+1} \ge \frac{w_j}{\alpha(k_j+1)}$.

We conclude that $w_{s+1} \ge \frac{w_1}{\alpha^s\prod_{j=1}^s (k_j+1)} = \frac{N}{\alpha^s\prod_{j=1}^s (k_j+1)}$. Also, it is obvious that the adversary stops when $w_{s+1} \leq 1$. Hence $\frac{N}{\alpha^s\prod_{j=1}^s (k_j+1)} \le w_{s+1} \le 1$ and since $(1+x) \le e^x$ for all $x \in \R$, it holds that 
$$N \le \alpha^s\prod_{j=1}^s e^{k_j} = e^{s\log\alpha + \sum_{j=1}^s k_j}.$$ 
Now, we apply $\log$ to this equation, and we obtain $\log N \le s\log\alpha + \sum_{j=1}^s k_j$. Also, recall that for each $j$, $1 \le j \le s$, $O_j$ is not hit at the time of its arrival so we have $k_j \ge 1$, which implies $s \le \sum_{j=1}^s k_j$. Therefore, we obtain $\frac{\log N}{1 + \log\alpha} \le \sum_{j=1}^s k_j$ which finishes the proof.
\end{proof}

\begin{figure}[h!]
    \centering
    \includegraphics[page=6, width=0.5\textwidth]{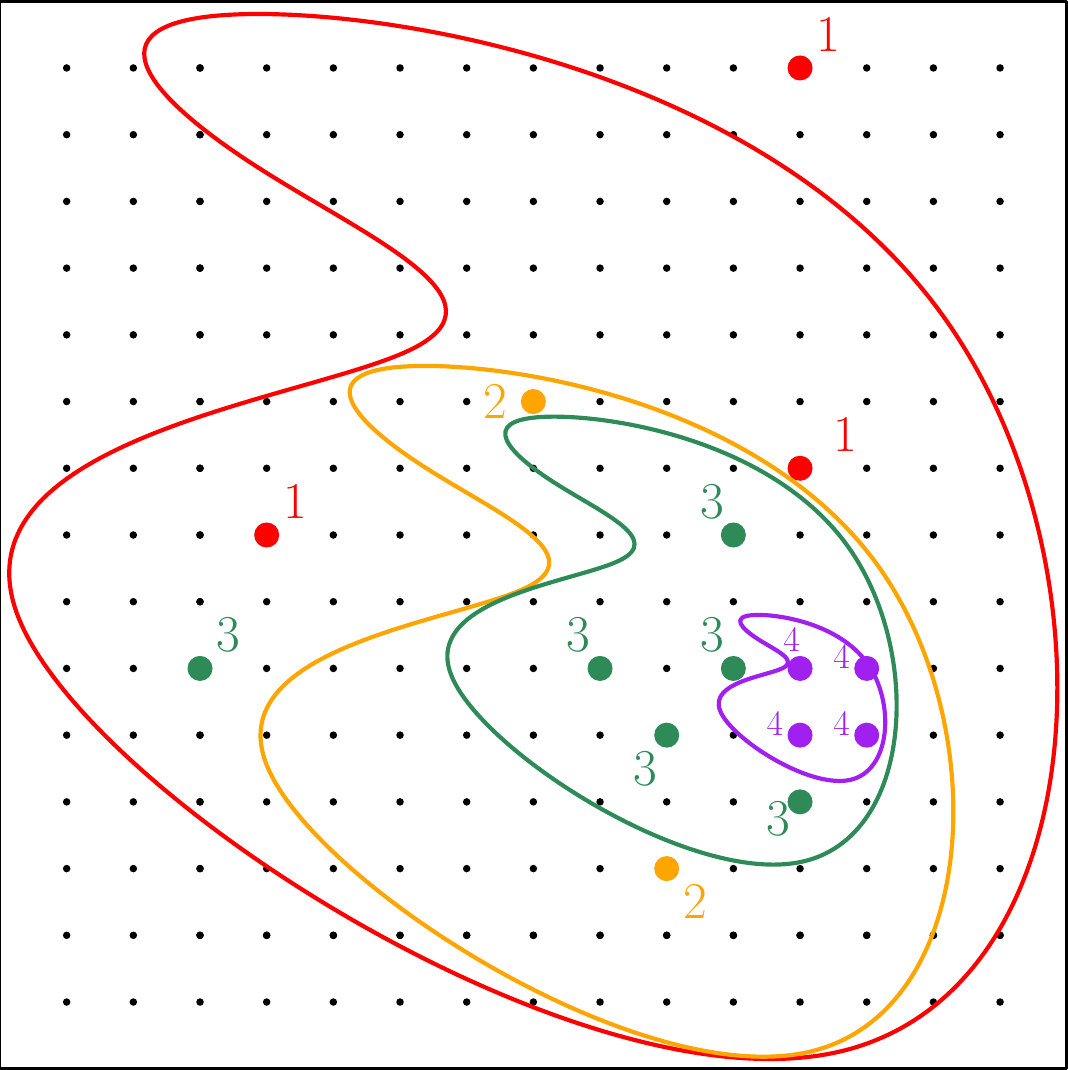}
    \caption{This figure illustrates the strategy of the adversary. In each turn, the adversary finds the biggest dilation of $O$ that is not hit by any point inside the previous object. The points with the same number are the set of points chosen by the algorithms in each step.}
    \label{fig:adversary_strategy}
\end{figure}

\begin{figure}[h!]
    \centering
    \includegraphics[page=7, width=0.5\textwidth]{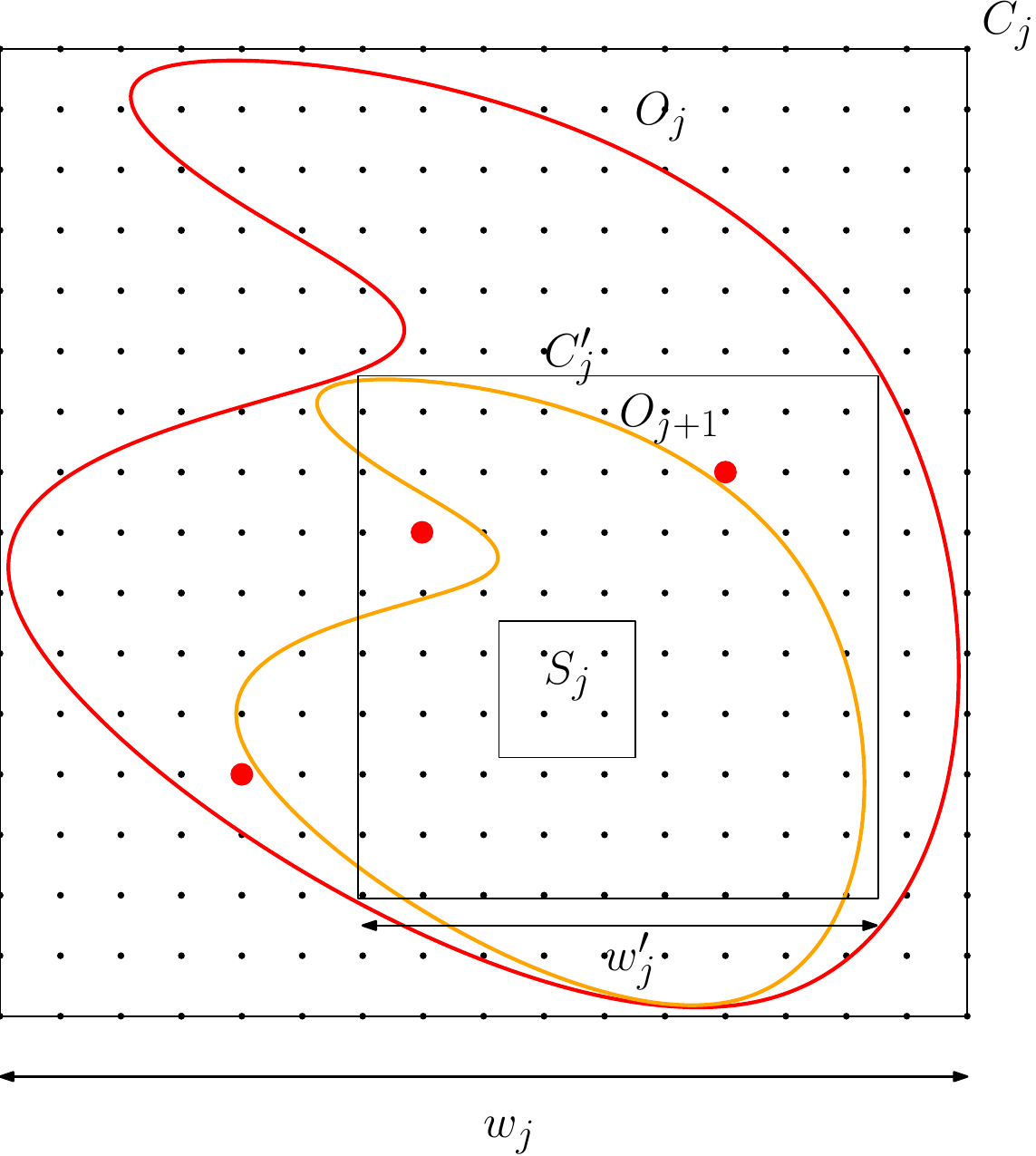}
    \caption{This figure shows the existence of a $d$-cube of width at least $\frac{w_j}{\alpha(k_j+1)}$ at the end of $j$-th step of the game that does not contain any point of the algorithm. Red points are the points chosen by the algorithm in the $j$-th step.}
    \label{fig:lower_bound}
\end{figure}

\section{Conclusion}

We have presented a tight $O(\log n)$-competitive algorithm for the online hitting set problem of fat objects of fixed dimension and fixed aspect ratio, when the set of points corresponds to all the points of integral coordinates contained in a fixed hypercube of width $N=n^{1/d}$. We finish with some open questions. 
\begin{itemize}
    \item For any $d$, when the algorithm is allowed to use any point in $\Z^d \cap (0,N)^d$, can one either design an online algorithm for $\alpha$-\emph{fat} objects with a better competitive ratio than $((4\alpha+1)^{2d}\log N)$ or improve the lower bound of $\frac{\log N}{1 + \log \alpha}$ to tighten the gap?
    \item When $d=2$, and the algorithm is allowed to use any point in a fixed square of width $N$, can one design an online algorithm for \emph{rectangles} with competitive ratio $O(\log N)$? 
    \item When $d=2$, and the set of points is a fixed subset (known in advance)
    $P\subseteq \Z^2\cap (0,N)^2$, can one design an online algorithm for \emph{disks} with competitive ratio $O(\log N)$? 
     \item When $d=3$, and the set of points is a fixed subset (known in advance) 
    $P\subseteq \Z^3\cap (0,N)^3$, can one design an online algorithm for $3$-\emph{cubes} with competitive ratio $O(\log N)$? 
    \item When $d=2$, and the set of points is a fixed subset (known in advance) 
    $P\subseteq \R^2$ of size $n$, can one design an online algorithm for \emph{squares} with competitive ratio $O(\log n)$? 
\end{itemize}

 \bibliographystyle{splncs04}
 \bibliography{bibliography}

\end{document}